\documentclass{birkjour}

\newtheorem{thm}{Theorem}[section]

\usepackage{amsmath}
\usepackage{amsfonts}
\usepackage{amssymb}
\usepackage{graphicx}
\usepackage{amsfonts}
\usepackage[english]{babel}
\usepackage[utf8]{inputenc}
\usepackage{times}
\usepackage[T1]{fontenc}

\begin{document}

\title[Revisiting EPRL: All Finite-Dimensional Solutions]{Revisiting EPRL: All Finite-Dimensional Solutions by Naimark's Fundamental Theorem}

\author{Leonid Perlov}
\address{Department of Physics, University of Massachusetts,  Boston}
\email{leonid.perlov@umb.edu}
\author{Michael Bukatin}
\address{HERE North America LLC, Burlington, MA, USA}
\email{bukatin@cs.brandeis.edu}

\date{ March 26, 2016}

\begin{abstract}
In this paper we research all possible finite-dimensional representations and corresponding values of the Barbero-Immirzi parameter contained in EPRL simplicity constraints by using Naimark's fundamental theorem of the Lorentz group representation theory. It turns out that for each non-zero pure imaginary with rational modulus value of the Barbero-Immirzi parameter $\gamma = i \frac{p}{q}, p, q \in Z, p, q \ne 0$, there is a solution of the simplicity constraints, such that the corresponding Lorentz representation is finite-dimensional. The converse is also true - for each finite-dimensional Lorentz representation solution of the simplicity constraints $(n, \rho)$, the associated Barbero-Immirzi parameter is non-zero pure imaginary with rational modulus, $\gamma = i \frac{p}{q}, p, q \in Z, p, q \ne 0$. We solve the simplicity constraints with respect to the Barbero-Immirzi parameter and then use Naimark's fundamental theorem of the Lorentz group representations to find all finite-dimensional representations contained in the solutions. 
\end{abstract}

\maketitle

\section{Introduction}
Since we are researching the finite dimensional representations we will be using the original form of EPRL~\cite{EPRL} rather than the latest form~\cite{LinearizedEPRL}, which provides solutions only for the unitary principal series representation $(j, \gamma j), j \in Z, \gamma \in R$. Following~\cite{SU(2)Projection} we assume that the Barbero-Immirzi parameter is not a constant and can take any complex value. The most recent attempt to research the finite dimensional representations was made in~\cite{Wieland1,Wieland2}. 

The original form of the EPRL constraints~\cite{EPRL,Pereira2008,FlippedVertex,VertexAmplitude} contains the diagonal and non-diagonal constraints. The diagonal simplicity constraint provides two series of solutions: $\rho = n \gamma$ and  $\rho = -n/{\gamma}$, where $(n, \rho)$ are the Lorentz group principal series representation parameters, $\gamma$ is the Barbero-Immirzi parameter~\cite{Barbero}. Only the unitary infinite dimensional solutions of the first series were selected and researched in the original EPRL paper~\cite{EPRL} and in the new formulation of the Loop Quantum Gravity~\cite{RovelliBook,RovelliBook2,Rovelli2010,Thiemann}. The second series was always rejected and not researched on the ground of providing the complex values of the Barbero-Immirzi parameter, believed to have no physical meaning. In this paper we formally research the finite dimensional representations and the corresponding values of  Barbero-Immirzi parameter of both solution series by using Naimark's fundamental theorem of the Lorentz group representation theory.

By solving the off-diagonal simplicity constraint $n\rho = 4\gamma L^2$, ($L$ is an $SU(2)$ generator)~\cite{EPRL}, with respect to $\gamma$, we find that the first series contains the solutions $(n, \rho = n \gamma )$, and $\gamma$ can take any complex value,  while the second series solutions provides the following values for $\gamma =  \pm i (n/(|n|+2r)), \; n \in Z, n \ne 0, r= 0, 1, 2 \mbox{ ... }$ and the Lorentz representation parameters: $(n, \; \rho= \pm i(|n| + 2r))$. By using Naimark's fundamental theorem from (~\cite{Naimark} p. 295) we show that all second series solutions are necessary finite dimensional. More than that, we show that all finite dimensional representations contained in the first series are the same as in the second:  $(n, \; \rho= \pm i(|n| + 2r))$, but correspond to the inverse values of $\gamma$: $\gamma =  \pm i ((|n|+2r)/n)$.  The only common values of $\gamma$ for both series are $\gamma = \pm i$. For this special case we prove by using Naimark's fundamental theorem that if $\gamma= \pm i$ then the corresponding solutions are necessary finite dimensional and therefore non-unitary. Finally we prove the Main Theorem stating that there is a correspondence between the pure imaginary with rational modulus values of Barbero-Immirzi parameter and the finite dimensional Lorentz representation solutions of the simplicity constraints. For each non-zero pure imaginary with rational modulus value of the Barbero-Immirzi parameter $\gamma = i \frac{p}{q}, p, q \in Z, p, q \ne 0$, there is a solution of the simplicity constraints, such that the corresponding Lorentz representation is finite dimensional. The converse is also true - for each finite dimensional Lorentz representation solution of the simplicity constraints $(n, \rho)$, the associated Barbero-Immirzi parameter is non-zero pure imaginary with rational modulus, $\gamma = i \frac{p}{q}, p, q \in Z, p, q \ne 0$. 

The paper is organized as follows. In section \ref{sec:EPRL SimplicityConstraints} we review the EPRL form of the simpicity constraints and the two series of the diagonal constraint solutions.  In section \ref{sec:firstseries} we solve the off-diagonal constraint for the first series and apply the Naimark's fundamental theorem to find all finite dimensional representation contained in it and the corresponding values of the Barbero-Immirzi parameter. In section \ref{sec:secondseries} we find the solutions of the off-diagonal constraint for the second series and the corresponding values of $\gamma$ and by using the fundamental theorem show that all those solutions are finite dimensional. In section \ref{sec:specialcase} we consider the special case of the $\gamma =\pm i$ and prove that the simplicity constraints and the fundamental theorem require that the corresponding Lorentz representations be necessary finite dimensional and therefore non-unitary.  In section \ref{sec:MainTheorem} we prove the Main Theorem identifying all possible Barbero-Immirzi parameter values  corresponding to all possible finite dimensional Lorentz representation simplicity constraints solutions.  The discussion section \ref{sec:Discussion} completes the paper.
\section{EPRL Simplicity Constraints}
\label{sec:EPRL SimplicityConstraints}
The diagonal and off-diagonal simplicity constraints in EPRL form~\cite{EPRL} are
\begin{equation}
\label{diagonal}
C_2\left(1-\frac{1}{\gamma^2}\right) + \frac{2}{\gamma}C_1 = 0
\end{equation}
\begin{equation}
\label{offdiagonal}
C_2=4\gamma L^2
\end{equation}
,where $C_1 = J \cdot J$ - Casimir scalar and $C_2 = \,^*J \cdot J$ - Casimir pseudo-scalar, L - rotation generators, $\gamma$ -Barbero-Immirzi parameter. And the equality in ($\ref{diagonal}$) is a weak equality: $<\psi \, C \,  \phi> = 0$\\

For the Lorentzian theory, when using the principal series representation (not necessary unitary), the Casimir $C_1$ and the Pseudo-Casimir $C_2$ can be expressed by using the representation parameters $(n, \rho)$, where $n \in Z, \; \rho \in C$ as follows~\cite{Ruhl,Knapp,Gelfand}:
\begin{equation}
\label{Casimir}
C_1 = \frac{1}{2}(n^2  - {\rho}^2 -4)
\end{equation}
\begin{equation}
\label{PseudoCasimir}
C_2 = n\rho
\end{equation}
or if one uses a different ordering, i.e. selects the spectrum $j^2$ instead of $j(j+1)$ for the angular momentum~\cite{EPRL}, the Casimir $C_1$ takes the form (see Appendix A):
\begin{equation}
\label{Casimir1}
C_1 = \frac{1}{2}(n^2  - {\rho}^2)
\end{equation}
By substituting ($\ref{Casimir1}$) and ($\ref{PseudoCasimir}$) into ($\ref{diagonal}$) and ($\ref{offdiagonal}$) we obtain:
\begin{equation}
n \rho\left ( \gamma -\frac{1}{\gamma} \right ) = {\rho}^2 - n^2
\end{equation}
and
\begin{equation}
\label{offdiagonal1}
n \rho =4\gamma L^2
\end{equation}
The first, diagonal constraint, produces the two series of solutions:
\begin{equation}
\label{firstsolution}
\rho = n \gamma 
\end{equation}
and 
\begin{equation}
\label{secondsolution}
\rho = \frac{-n}{\gamma}
\end{equation}
Since EPRL paper~\cite{EPRL} considered only the unitary principal representations, i.e. $\gamma \in R$, only the first solution $\rho = n\gamma$ was selected, as the second solution was believed to be not suitable due to the assumption of $\gamma$ being real. As it was shown in the EPRL paper, for the first series $\rho =n\gamma$ the off-diagonal constraint selects the lowest $SU(2)$ representation in the direct sum decomposition $H_{(n, \rho)} = \bigoplus\limits_{m \ge |n|/2} H_m$ in case of real $\gamma$ and unitary representation.\\[2ex]
The second series of solutions  and the solutions of the first series corresponding to the complex-valued $\gamma$ were not researched. In this paper we do not assume  that $\gamma$ is real or that representations are unitary. Let us see in the next sections what one can obtain from the off-diagonal constraints ($\ref{offdiagonal}$) if one drops such assumptions.

\section{The First Series of  Solutions}
\label{sec:firstseries}
The first series of the solutions is  $\rho = n \gamma$. It is obtained from the diagonal constraint:
\begin{equation}
\label{diagonal2}
n \rho\left ( \gamma -\frac{1}{\gamma} \right ) = {\rho}^2 - n^2
\end{equation}
,where $ n \in Z, \rho \in C$. We consider the case $n \ne 0$. If $n = 0$, if follows from $\rho = n \gamma$ that $\rho = 0$ for any $\gamma$ and the solution $(n = 0, \rho = 0)$ corresponds to the trivial representation. 

Let us substitute the diagonal constraint solution $\rho = n \gamma$ into the non-diagonal constraint:
\begin{equation}
\label{offdiagonal2}
n \rho =4\gamma L^2
\end{equation}
and instead of $SU(2)$ generator L take its values $(\frac{|n|}{2} + s)$, where $n \in Z, n \ne 0, s = 0, 1, \mbox{...}$, then we obtain:
\begin{equation}
\label{offdiagonal3}
n^2 \gamma = 4 \gamma {\left (\frac{|n|}{2} + s \right)}^2
\end{equation}
$\gamma$ cancels on both sides and one can see that the solution is $s=0$. In other words the off diagonal constraint does not provide any restrictions on $\gamma$. The final solution of the both diagonal and off-diagonal constraints can be written as $(n, \rho = n \gamma) $, where n is an integer, $\gamma$ - any complex number. Let us find all finite dimensional solutions contained in $(n, \rho = n \gamma)$. \\[2ex]
Naimark's fundamental theorem of the Lorentz group representations theory (~\cite{Naimark} page 295) states that the principal series representation $(n, \rho)$ is finite dimensional if and only if  
\begin{equation}
\label{Maintheorem5}
\exists \;  r  = 1, 2 \mbox{...} \; \; \mbox{ , so that } \;  {\rho}^2 = - {(|n|+ 2r)}^2
\end{equation}
Since we used the Casimir from EPRL:  $C_1 = \frac{1}{2}(n^2  - {\rho}^2)$ instead of $C_1 = \frac{1}{2}(n^2  - {\rho}^2 - 4)$ corresponding to the $j^2$ spectrum instead of $j(j+1)$ we have to rewrite this condition for EPRL ordering
\begin{equation}
\label{Maintheorem4}
\exists \;  r  = 0, 1, 2 \mbox{...} \; \; \mbox{  , so that } \;  {\rho}^2 = - {(|n| + 2r)}^2
\end{equation}
It is hard to notice the difference with ($\ref{Maintheorem5}$), however $r$ now takes values starting with zero rather than one as in ($\ref{Maintheorem5}$) (see the Appendix A for details).
%\textbf{Theorem 1:} 
\begin{thm}
All first series finite dimensional non-trivial representation solutions  are of the form:
$\gamma = \pm i \left ( \frac{|n| + 2r}{n} \right )$ , $(n, \rho = \pm i(|n| + 2r)), n \in Z, n \ne 0, r = 0, 1, \mbox{ ... }$
\end{thm}

%\textbf{Proof:}\\[2ex]
\begin{proof}
By substituting the first EPRL series solution $(n, \rho = n \gamma )$ into the finite dimensionality condition ($\ref{Maintheorem4}$) of the Naimark's fundamental theorem, we obtain:
\begin{equation}
{(n \gamma ) }^2 = -{(|n| + 2r)}^2
\end{equation}
which provides the following solution, when solved with respect to $\gamma$:
\begin{equation}
\label{gammasolution}
\gamma = \pm i \left ( \frac{|n| + 2r}{n} \right )
\end{equation}
where, $n \in Z, n \ne 0 $, $r = 0, 1, 2, \; \mbox{...}$\\
By substituting the found $\gamma$ into $\rho = n \gamma$, we find:
\begin{equation}
\rho = \pm i ( |n| + 2r) 
\end{equation}
hence we find that  
\begin{equation}
\label{firstseries6}
\forall r = 0, 1, \mbox{...} \; \;  \gamma = \pm i \left ( \frac{|n| + 2r}{n} \right ),  \; \; (n, \rho = \pm i ( |n| + 2r) )
\end{equation}
satisfy the Naimark's fundamental theorem condition.
\end{proof}
%$\square$ \\[2ex]
We note that $|\gamma| \ge 1$ and $\gamma = \pm i$, when $r=0$.

\section{The Second Series of  Solutions}
\label{sec:secondseries}
The second series solutions of the diagonal constraint:
\begin{equation}
\label{diagonal3}
n \rho\left ( \gamma -\frac{1}{\gamma} \right ) = {\rho}^2 - n^2
\end{equation}
is  $\rho = -n/\gamma$ \\[2ex]
This solution was not researched before as it provides the pure imaginary values for $\gamma$. Let us consider again the non-diagonal constraint:
\begin{equation}
\label{offdiagonal35}
n \rho =4\gamma L^2
\end{equation}
or, when substituting the parameters of $SU(2)$ generator L: $(\frac{|n|}{2} + s)$, where $s = 0, 1, \mbox{...}$ 
\begin{equation}
\label{offdiagonal4}
n \rho = 4 \gamma {\left (\frac{|n|}{2} + s \right)}^2
\end{equation}
Again we assume $n \ne  0$, otherwise we again get the trivial representation solution $(n=0, \rho = 0)$\\

After substituting $\rho = -n/\gamma$ into ($\ref{offdiagonal4}$) we obtain:
\begin{equation}
\frac{-n^2}{\gamma} = 4 \gamma {\left (\frac{|n|}{2} + s \right)}^2
\end{equation}
 when solving with respect to $\gamma$
\begin{equation}
\label{gammasolution1}
\gamma = \pm i \left ( \frac{n}{|n| + 2s} \right )
\end{equation}
where $n \in Z, n \ne 0 $, $s = 0, 1, 2 ...$ 
The values $\gamma = \mp i $ are achieved, when $s=0$. We also note that all other $\gamma$ values are pure imaginary and $|\gamma| \le 1$. It tends to zero, when $s \rightarrow \infty$.\\[2ex]
By substituting the found solution ($\ref{gammasolution1}$)  for $\gamma$  into  $\rho = -n/\gamma$ we obtain:
\begin{equation}
\rho = \pm i (|n| + 2s),  \mbox{, where   }  s = 0, 1, 2, \mbox{...}
\end{equation}
Thus the solution is
\begin{equation}
\label{newlyfoundsolution}
\gamma = \pm i \left ( \frac{n}{|n| + 2s} \right )\; , \;\;   ( \; n,  \; \rho = \pm i(|n|+2s) ), \; n \in Z
\end{equation}
We would like to emphasize that in the second series the values for $\gamma$ were obtained by solving the off-diagonal simplicity constraint, while in the first series $\gamma$ canceled on both sides of the off-diagonal constraint allowing $\gamma$ to take any complex value.
\begin{thm} 
All second series solutions correspond to the Lorenz group finite dimensional representations.
\end{thm}
%\textbf{Proof:}\\[2ex]
\begin{proof}
The second series solutions are of the form:
\begin{equation}
\label{newlyfoundsolution2}
\gamma =\pm  i \left ( \frac{n}{|n| + 2s} \right ), \;\;  (n, \rho = \pm i (|n|+2s) )
\end{equation}
The fundamental theorem of the Lorentz group representations theory (~\cite{Naimark} page 295) says that the principal series representation $(n, \rho)$ is finite dimensional if and only if  
\begin{equation}
\label{Maintheorem}
\exists \;  r  = 1, 2 \mbox{...} \; \; \mbox{ , so that } \;  {\rho}^2 = - {(|n|+ 2r)}^2
\end{equation}
When one selects the spectrum $j^2$ instead of $j(j+1)$ which causes the Casimir to become $C_1 = \frac{1}{2}(n^2  - {\rho}^2)$ instead of $C_1 = \frac{1}{2}(n^2  - {\rho}^2 - 4)$ this condition can be rewritten as (see the Appendix A):
\begin{equation}
\label{Maintheorem1}
\exists \;  r = 0, 1, 2 \mbox{...} \; \; \mbox{  , so that } \;  {\rho}^2 = - {(|n| + 2r)}^2
\end{equation}
where $r$ now takes values starting with zero.\\
By substituting ($\ref{newlyfoundsolution2}$) into ($\ref{Maintheorem1}$), for any $s= 0, 1, \mbox{...}$, we take $r=s$ and obtain:
\begin{equation}
-(|n|+2s)^2 = - {(|n| + 2r)}^2
\end{equation}
Thus we can see that all second series solutions satisfy condition ($\ref{Maintheorem1}$).
Therefore all newly found solutions correspond to the Lorentz group finite dimensional representations.
\end{proof}
%$\square$\\[2ex]
One can see that in both first and second series the finite dimensional representation solutions are the same: 
$(n, \pm i(|n| + 2r))$, however they correspond to the different values of Barbero-Immirzi parameter $\gamma$. In fact the values of the first series are the inverse of the second series values as one can see from the expressions ($\ref{firstseries6}$) and ($\ref{newlyfoundsolution}$). It is also interesting to notice that the second series contains only finite dimensional representation, while the first one contains both finite and infinite representations. The only common values of the both series are $\gamma=\pm i$. It is also interesting that while the second series solutions provide the values of Barbero-Immirzi $\gamma$  as the solution of the off-diagonal simplicity constraint, the first series does not provide any values for $\gamma$ as $\gamma$ cancels on both sides of the off-diagonal constraint. The values of $\gamma$ corresponding to the finite dimensional representations of the first series are obtained from the Naimark's fundamental theorem. On the contrary in the case of the second series the fundamental theorem does not select any subset of the $\gamma$ values, stating that they all are finite dimensional. \\[2ex]

\section{The Special Case of the Barbero-Immirzi $\gamma = \pm i $}
\label{sec:specialcase}
Since the value of $\gamma = \pm i$ is of a special interest and correspond physically to the self-dual and anti-self dual connections, we would like to formally prove the following theorem.
%\textbf{Theorem 3:}
\begin{thm}
If the Barbero-Immirzi parameter $\gamma = \pm i$ then the corresponding Lorentz representation solutions of the simplicity constraints are necessary finite dimensional and therefore non-unitary.
\end{thm}
%\textbf{Proof:}\\
\begin{proof}
If $\gamma = \pm i$ then both the first diagonal constraint solution $\rho = n \gamma$ and the second $\rho = -n/\gamma$ provide the same solution $\rho = \pm n i$, which corresponds to the principal non-unitary representations $(n, \pm n i )$. By the Lorentz group representation theory fundamental theorem (~\cite{Naimark} p.295)  these representations are necessary finite dimensional since they satisfy the representation finite dimensionality condition for $r=0$: $\exists \;  r  = 0, 1, 2 \mbox{...} \; \; \mbox{  , so that } \;  {\rho}^2 = - {(|n| + 2r)}^2$.
\end{proof}
%$\square$

\section{The Main Theorem}
\label{sec:MainTheorem}
%\textbf{The Main Theorem:} 
\begin{thm}[The Main Theorem]
Barbero-Immirzi parameter $\gamma$ is non-zero pure imaginary with rational modulus value, i.e. of the form: $\gamma = i\frac{p}{q}, p, q \ne 0$ if and only if $\gamma$ is a solution of the simplicity constraints, such that the Lorentz group representation corresponding to this solution is finite dimensional. These Lorentz group representations are described as follows: \\
If $|\gamma| \ge 1$, that is  $|p| \ge |q|$ and (($p >0$ and $k >0$)  or ($p < 0$ and $k<0$)),\\
then $n = \pm 2qk, \; \rho = \pm 2ipk$;\\[2ex]
If $|\gamma| \le 1$, that is  $|p| \le |q|$ and (($q >0$ and $k >0$)  or ($q < 0$ and $k<0$)),\\
then $n = \pm 2pk, \; \rho = \pm 2iqk$;\\[2ex]
where $\; n, k \in Z,\; n, k \ne 0$.
\end{thm}
%Case: $|\gamma| \ge 1$, that is  $|p| \ge |q|$, $\; n, k \in Z,\; n, k \ne 0$\\[2ex]
%when  $p > 0$ and $ q > 0$ \\
%then $k > 0, \; \gamma = i \frac{p}{q}, \; n = \pm 2qk, \; \rho = \pm 2ipk$ \\
%when $p < 0$ and $ q < 0$ \\
%then $k < 0, \; \gamma = i \frac{p}{q}, \; n = \pm 2qk, \; \rho = \pm 2ipk$ \\
%when $p > 0$ and $ q < 0$ \\
%then $k > 0, \; \gamma = i \frac{p}{q}, \; n = \mp 2qk, \; \rho = \mp 2ipk$  \\
%when $p < 0$ and $ q > 0$ \\
%then $k < 0, \; \gamma = i \frac{p}{q}, \; n = \mp 2qk, \; \rho = \mp 2ipk$  \\[2ex]
%Case: $|\gamma| \le 1$, that is  $|p| \le |q|$, $\; n, k \in Z, \; n, k \ne 0$ \\[2ex]
%when $p > 0$ and $ q > 0$ \\
%then $k > 0, \; \gamma = i \frac{p}{q}, \; n = \pm 2pk, \; \rho = \pm 2iqk$ \\
%when $p < 0$ and $ q < 0$ \\
%then $k < 0, \; \gamma = i \frac{p}{q}, \; n = \pm 2pk, \; \rho = \pm 2iqk$ \\
%when $p < 0$ and $ q > 0$ \\
%then $k > 0, \; \gamma = i \frac{p}{q}, \; n = \mp 2pk, \; \rho = \mp 2iqk$\\
%when $p > 0$ and $ q < 0$ \\
%then $k < 0, \; \gamma = i \frac{p}{q}, \; n = \mp 2pk, \; \rho = \mp 2iqk$ \\[2ex] 
%\textbf{Proof:}\\
\begin{proof}
The proof in one direction is straightforward. It follows from the Theorem 1 and Theorem 2 that all finite dimensional solutions of the simplicity constraints of the first series are:
\begin{equation}
\gamma =\pm  i \left ( \frac{|n| + 2r}{n} \right ), \; \;  (n, \rho = \pm i(|n| + 2r) )
\end{equation}
and for the second are:
\begin{equation}
\gamma =\pm  i \left ( \frac{n}{|n| + 2r} \right ), \; \;  (n, \rho = \pm i(|n| + 2r) )
\end{equation}
where, $ n \in Z, n \ne 0 , r = 0, 1, \mbox{...} $
We can immediately see that the values of $\gamma$ in both cases are non-zero pure imaginary with the rational modulus. \\[2ex]
Now let us assume that $\gamma$ is non-zero pure imaginary with the rational modulus, i.e. of the form  $ i \frac{p}{q} $ , where $p, q \in Z, p, q \ne  0 $  and find all finite dimensional representation solutions of the simplicity constraints corresponding to it. 
For convenience let us consider non reducible $\frac{p}{q}$ and multiply by 2 the numerator and denominator  $ \frac{2p}{2q} i$, and find n and r for each simplicity constraints solution, when expressed via p and q. We will consider the following eight cases that cover all situations.  The two first series solutions $ \gamma = \pm i \left ( \frac{|n| + 2r}{n} \right )$ with plus and minus $i$ in cases of $n > 0$ and $n < 0$. This will provide four cases. The same for the second series solutions we consider four cases: $ \gamma =   \pm i \left ( \frac{n}{|n| + 2r} \right )$ for $\gamma$ with $\pm i$ for $n >0$ and $n < 0$. Altogether there are eight cases. \\[2ex]
\textbf{Case 1:}  The first series of solutions  with plus $i$ and $n > 0$ : $ \gamma = i \left ( \frac{|n| + 2r}{n} \right )$
\begin{equation}
\gamma = i \frac{2p}{2q} =  i \frac{n + 2r}{n}
\end{equation}
We find 
\begin{equation}
n = 2qk,  \; \; r = (p - q)k, \; k \in Z, k \ne 0 
\end{equation}
Since $ n \in Z, n > 0, \; r = 0, 1 \mbox{...} $ it follows that either
\begin{equation}
\mbox{Case 1a:  } k > 0, \; q > 0, \; p > 0, \;  |p| \ge |q|
\end{equation}
or
\begin{equation}
\mbox{Case 1b:  } k < 0, \; q < 0, \;  p < 0, \; |p| \ge |q|
\end{equation}
We explain the logic for the Case1a, all other cases are very similar:\\
It is easy to see, that since in Case1a $n > 0$ and $k > 0$ then from $n =2qk$ it follows that $q > 0$. Then, since $r = 0, 1, \mbox{...}$, it follows from $r = (p - q)k$,  that $p > 0$ and $ |p| \ge |q|$ Also it follows from $r = (p - q)k$ that $k$ should be a whole number, since $r$ is a whole number and $p$ and $q$ are mutually prime. \\[2ex]
\textbf{Case 2:}  The first series of solutions with  $-i$ and $n>0$:  $ \gamma =  - i \left ( \frac{|n| + 2r}{n} \right )$
\begin{equation}
\gamma = i\frac{2p}{2q} =  -i \frac{n + 2r}{n}
\end{equation}
We find 
\begin{equation}
n = - 2qk,  \; \; r = (p + q)k
\end{equation}
Since $ n \in Z, n > 0, \; r = 0, 1 \mbox{...} $ it follows that either
\begin{equation}
\mbox{Case 2a:  } k > 0, \; q < 0, \; p > 0;  |p| \ge |q|
\end{equation}
or
\begin{equation}
\mbox{Case 2b:  } k < 0, \;  q > 0,\; p < 0, \;  |p| \ge |q|
\end{equation}
\textbf{Case 3:}  The first series of solutions  with plus $i$ and $n < 0$ : $ \gamma = i \left ( \frac{|n| + 2r}{n} \right )$
\begin{equation}
\gamma = i\frac{2p}{2q} =  i\frac{2r-n}{n}
\end{equation}
We find 
\begin{equation}
n = 2qk,  \; \; r = (p + q)k
\end{equation}
Since $ n \in Z, n < 0, \; r = 0, 1 \mbox{...} $ it follows that either
\begin{equation}
\mbox{Case 3a:  } k > 0, \;  q < 0, \;  p > 0, \;   |p| \ge |q|
\end{equation}
or
\begin{equation}
\mbox{Case 3b:  } k < 0, \; q > 0, \;  p < 0, \; |p| \ge |q|
\end{equation}
\textbf{Case 4:} The first series of solutions with  $-i$ and $n<0$:  $ \gamma =  - i \left ( \frac{|n| + 2r}{n} \right )$
\begin{equation}
\gamma = i\frac{2p}{2q} =  -i \frac{2r-n}{n}
\end{equation}
We find 
\begin{equation}
n = -2qk,  \; r = (p - q)k
\end{equation}
Since $ n \in Z, n < 0, \; r = 0, 1 \mbox{...} $ it follows that either
\begin{equation}
\mbox{Case 4a:  } k > 0, \; q > 0, \; p > 0, \; |p| \ge |q|
\end{equation}
or
\begin{equation}
\mbox{Case 4b:  } k < 0, \;  q < 0, \; p < 0, \; |p| \ge |q| 
\end{equation}
\textbf{Case 5:}  The second series of solutions  with $-i$ and $n> 0$ : $ \gamma =  - i \left ( \frac{n}{|n| + 2r} \right )$
\begin{equation}
\gamma = i\frac{2p}{2q} =  -i \frac{n}{n + 2r}
\end{equation}
We find 
\begin{equation}
n = -2pk,  \; \; r = (p + q)k
\end{equation}
Since $ n \in Z, n > 0, \; r = 0, 1 \mbox{...} $ it follows that either
\begin{equation}
\mbox{Case 5a:  } k > 0, \; p < 0 , \; q > 0,\; |p| \le |q|
\end{equation}
or
\begin{equation}
\mbox{Case 5b:  } k  < 0, \; p > 0, \; q < 0, \; |p| \le |q|
\end{equation}
\textbf{Case 6:}  The second series of solutions  with $i$ and $n > 0$: $ \gamma =  i \left ( \frac{n}{|n| + 2r} \right )$
\begin{equation}
\gamma = i\frac{2p}{2q} =  i\frac{n}{n + 2r}
\end{equation}
We find 
\begin{equation}
n = 2pk,  \; \; r = (q - p)k 
\end{equation}
Since $ n \in Z, n > 0, \; r = 0, 1 \mbox{...} $ it follows that either
\begin{equation}
\mbox{Case 6a:  } k > 0, \; p > 0 , \; q > 0, \; |p| \le |q|
\end{equation}
or
\begin{equation}
\mbox{Case 6b:  } k < 0, \; p < 0 , \; q < 0, \; |p| \le |q|
\end{equation}
\textbf{Case 7:}  The second series of solutions  with $-i$ and $n < 0$: $ \gamma =  - i \left ( \frac{n}{|n| + 2r} \right )$
\begin{equation}
\gamma = i\frac{2p}{2q} =  -i \frac{n}{2r - n}
\end{equation}
We find 
\begin{equation}
n = -2pk,  \; \; r = (q - p)k
\end{equation}
Since $ n \in Z, n < 0, \; r = 0, 1 \mbox{...} $ it follows that either
\begin{equation}
\mbox{Case 7a:  } k > 0, \; p > 0 ,  \; q > 0, \; |p| \le |q|
\end{equation}
or
\begin{equation}
\mbox{Case 7b:  } k < 0, \; p < 0 ,  \; q < 0, \; |p| \le |q|
\end{equation}
\textbf{Case 8:}  The second series of solutions  with $i$ and $n<0$: $ \gamma =   i \left ( \frac{n}{|n| + 2r} \right )$
\begin{equation}
\gamma = i\frac{2p}{2q} =  i\frac{n}{2r-n}
\end{equation}
We find
\begin{equation}
n = 2pk,  \; \; r = (q + p)k
\end{equation}
Since $ n \in Z, n < 0, \; r = 0, 1 \mbox{...} $ it follows that either
\begin{equation}
\mbox{Case 8a:  } k > 0, \; p < 0 , q > 0,  |p| \le |q|
\end{equation}
or
\begin{equation}
\mbox{Case 8b:  } k < 0, \; p > 0 , q < 0,  |p| \le |q|
\end{equation}
As we can see the Cases 1-4 cover the values of |p| > |q| for all combinations of p and q. This is expected as $|\gamma| \ge 1$ for the first series as it follows from Theorem 1. The Cases 5-8 cover the values $|q|>|p|$ for all combinations of p and q. It is also expected as $|\gamma| \le 1$ for the second series as it follows from Theorem 2. \\[2ex]
We rewrite it once again in a different form to show that for each pair $(p, q)$, where $p, q \in Z, p, q \ne 0$ one finds the $\gamma, k, n, \rho$, to be the parameters of the corresponding finite dimensional solutions of the simplicity constraints. \\[2ex]
By getting the solution values for $n$ from the Cases 1-8 and by recalling that for the Cases 1- 4 $\rho = n\gamma$,  and for Cases 5-8 $\rho = -\frac{n}{\gamma}$, we can express $\rho$ via $q, p$ and $k$: \\[2ex]
$|\gamma| \ge 1$, that is  $|p| \ge |q|$, $\; n, k \in Z, n, k \ne 0 $ \\[2ex]
when  $p > 0$ and $ q > 0$ \\
then $k > 0, \; \gamma = i \frac{p}{q}, \; n = \pm 2qk, \; \rho = \pm 2ipk$ (Case 1a, 4a)\\
when $p < 0$ and $ q < 0$ \\
then $k < 0, \; \gamma = i \frac{p}{q}, \; n = \pm 2qk, \; \rho = \pm 2ipk$ (Case 1b, 4b)\\
when $p > 0$ and $ q < 0$ \\
then $k > 0, \; \gamma = i \frac{p}{q}, \; n = \mp 2qk, \; \rho = \mp 2ipk$  (Case 2a, 3a)\\
when $p < 0$ and $ q > 0$ \\
then $k < 0, \; \gamma = i \frac{p}{q}, \; n = \mp 2qk, \; \rho = \mp 2ipk$  (Case 2b, 3b)\\[2ex]
$|\gamma| \le 1$, that is  $|p| \le |q|$,  $\;n, k \in Z, n, k \ne 0$\\[2ex]
when $p > 0$ and $ q > 0$ \\
then $k > 0, \; \gamma = i \frac{p}{q}, \; n = \pm 2pk, \; \rho = \pm 2iqk$ (Case 6a, 7a)\\
when $p < 0$ and $ q < 0$ \\
then $k < 0, \; \gamma = i \frac{p}{q}, \; n = \pm 2pk, \; \rho = \pm 2iqk$ (Case 6b, 7b)\\
when $p < 0$ and $ q > 0$ \\
then $k > 0, \; \gamma = i \frac{p}{q}, \; n = \mp 2pk, \; \rho = \mp 2iqk$ (Case 5a, 8a)\\
when $p > 0$ and $ q < 0$ \\
then $k < 0, \; \gamma = i \frac{p}{q}, \; n = \mp 2pk, \; \rho = \mp 2iqk$ (Case 5b, 8b)\\

We can see that the above cases can be written in the following compact form:\\
If $|\gamma| \ge 1$, that is  $|p| \ge |q|$ and (($p >0$ and $k >0$)  or ($p < 0$ and $k<0$)),\\
then $n = \pm 2qk, \; \rho = \pm 2ipk$;\\[2ex]
If $|\gamma| \le 1$, that is  $|p| \le |q|$ and (($q >0$ and $k >0$)  or ($q < 0$ and $k<0$)),\\
then $n = \pm 2pk, \; \rho = \pm 2iqk$;\\[2ex]
where $\; n, k \in Z,\; n, k \ne 0$.\\

We have proved that for each $p$ and $q$,  where $p, q  \in Z, p, q\ne 0$, 
 $\gamma = \pm i\frac{p}{q}$, there is a solution of the simplicity constraints, such that the corresponding Lorentz group representation is finite dimensional. And for each finite dimensional Lorentz representation solution of the simplicity constraint the corresponding $\gamma$ is the solution and is necessary of the form $i\frac{p}{q}$.
\end{proof}
%$\square$

\section{Discussion}
\label{sec:Discussion}
The main result of this paper is the Main Theorem stating that for each non-zero pure imaginary with rational modulus value of the Barbero-Immirzi parameter $\gamma = i \frac{p}{q}, p, q \in Z, p, q \ne 0$, there is a solution of the simplicity constraints, such that the corresponding Lorentz representation is finite dimensional. The converse is also true - for each finite dimensional Lorentz representation solution of the simplicity constraints $(n, \rho)$, the associated Barbero-Immirzi parameter is non-zero pure imaginary with rational modulus, $\gamma = i \frac{p}{q}, p, q \in Z, p, q \ne 0$.  \\[2ex]
In this paper we have found and researched all possible finite dimensional Lorentz representation solutions of the EPRL simplicity constraints. We used Naimark's fundamental theorem of the Lorentz representations to find all such representations. Instead of rejecting the second solution of the simplicity constraints $\rho = -n/\gamma$  on the ground of Barbero-Immirzi parameter being complex, we have researched and solved it with respect to the Barbero-Immirzi parameter. The result of the paper shows that the finite dimensional representations for both the first series $\rho = n \gamma$  and the second series: $\rho = -n/\gamma$ are the same:  $(n, \pm i(|n| + 2r))$ but correspond to the different values of $\gamma$. For the first series the  $\gamma$ values are $\gamma = \pm i \left ( \frac{|n| + 2r}{n} \right )$ , with $|\gamma| \ge 1$, while for the second the inverse of this expression: $\gamma = \pm i \left ( \frac{n}{|n| + 2r} \right )$ and $|\gamma| \le 1$. We have also proved in the Theorem 3, that if $\gamma = \pm i$, then the corresponding Lorentz group representation solutions $(n, \rho)$ are necessary finite dimensional and therefore non-unitary. The Main Theorem completes the paper.\\[2ex]

\section{Appendix A: Lorentz Group Representations Casimir and Pseudo-Casimir}
\label{sec:appendix_A}
The Lorentz group finite dimensional spinor representations are contained in the non-unitary principal series representations. They are usually parametrized by two half-integer spins $(j_-, j_+)$. 
Then the Casimir ($\ref{Casimir}$) and pseudo-Casimir ($\ref{PseudoCasimir}$), when expressed via the spins become:
\begin{equation}
\label{casimirs}
C_1 = 4(j_+( j_+ +1) + j_-( j_-+1))  \quad  C_2 = -4i( j_+(j_++1) - j_-(j_-+1) )  
\end{equation}
The  fundamental theorem in Naimark's book (`\cite{Naimark} p 295)  provides the expression for the connection between the spins $(j_+, j_-)$ and the principal series parameters $(n, \rho)$:
\begin{equation}
\label{spins2}
2j_+ = \frac{n}{2} + \frac{i \rho}{2} -1 \quad
2j_- =  -\frac{n}{2} + \frac{i \rho}{2} -1
\end{equation}
by adding and subtracting it follows that:
\begin{equation}
n = (2j_+ - 2j_-),   \quad
i\rho = ( 2j_+ + 2j_- +2)
\end{equation}
The condition for the representation being finite can be written following~\cite{Naimark} p 295 as:
\begin{equation}
\label{ChangingOrder1}
\rho = -i (|n| + 2r), \;\; r = 1, 2 \mbox{...}
\end{equation}
It is important to note that the values of $r$ begin with 1 rather than with zero, when one selects the $j(j+1)$ spectrum instead of $j^2$. 
The Casimir and Pseudo-Casimir are then as follows, which can be checked explicitly by using the expressions above:
\begin{equation}
C_1 = \frac{1}{2} \left ( n^2 - {\rho}^2 -4 \right ) = 4(j_+( j_+ +1) + j_-( j_-+1))
\end{equation}
\begin{equation}
C_2 = n\rho = -4i( j_+(j_++1) - j_-(j_-+1) )  
\end{equation}
When we select the spectrum as in \cite{EPRL} , i.e  $j^2$  instead of $j(j+1)$ all the formulas above change in the following way:\\
\begin{equation}
\label{casimirs1}
C_1 = 4({j_+}^2+{ j_-}^2)  \quad  C_2 = -4i( {j_+}^2 - {j_-}^2 )  
\end{equation}
\begin{equation}
\label{spins}
2j_+ = \frac{n}{2} + \frac{i \rho}{2} \quad
2j_- =  -\frac{n}{2} + \frac{i \rho}{2}
\end{equation}
where $\rho$ now for the finite dimensional representations is:
\begin{equation}
\label{ChangingOrder}
\rho = -i (|n| + 2r), \;\; r = 0, 1, 2 \mbox{...}
\end{equation}\\
It is very important to note the the values of $r$ now begin from zero, rather than from $1$ as it was in ($\ref{ChangingOrder1}$), when the spectrum was $j(j+1)$. The whole purpose of this Appendix is to show how the $r$ spectrum changes, when one changes the spin spectrum from $j(j+1)$ to $j^2$. \\[2ex]
From ($\ref{spins}$) one can also see that:
\begin{equation}
\label{nro}
n = (2j_+ - 2j_-),   \quad
i\rho = ( 2j_+ + 2j_-)
\end{equation}
and the Casimir and Pseudo-Casimir when expressed in $(n, \rho)$ become:
\begin{equation}
C_1  = 4 ({j_+}^2 + {j_-}^2) = \frac{1}{2} \left ( n^2 - {\rho}^2 \right )
\end{equation}
\begin{equation}
C_2  = -4i({j_+}^2 - {j_-}^2) = n\rho
\end{equation}
Of course the spectrums of $j_-$ and $j_+$ do not change when we change the ordering.  However the correspondence between $(j_-, j_+)$ and $(n, \rho)$ changes as it is seen from ($\ref{spins}$) and ($\ref{spins2}$). Particularly, when the ordering is $j^2$, the solution $(n = 0, \rho = 0)$ corresponding to $(j_-=0, j_+ = 0)$ is a finite dimensional of the dimension 1, corresponding to the trivial representation. At the same time, when the ordering is $j(j+1)$, the same solution $(j_- =0, j_+ = 0)$ corresponds to $(n=0, \rho = -2i)$.

\end{document}